%% file: SHQD.tex
\documentclass[11pt]{amsart}
\usepackage{geometry}                
\usepackage{graphicx}
\usepackage{epstopdf}
\usepackage[]{amsmath, amssymb, amsfonts, amsbsy, latexsym,amsthm}
\usepackage{enumerate}
\usepackage{graphicx}
\usepackage{color}
\newtheorem{theorem}{Theorem}[section]
\newtheorem{corollary}[theorem]{Corollary}

\newtheorem{thm}{Theorem}[section]

\newtheorem{prop}[thm]{Proposition}

\newtheorem{defn}[thm]{Definition}

\newtheorem{ex}[thm]{Example}

\DeclareMathOperator{\TR}{tr}

\newcommand{\UNFs}[1]{{\Omega}_{n,m} {(\mathbb {R})}}  
\newcommand{\Gr}[1]{\mu_{n,m} (\mathbb R)} 

\title{A short history of frames and quantum designs}
\author{Bernhard G. Bodmann and John I. Haas}

\begin{document}

\begin{abstract} In this survey, we relate frame theory and quantum information theory, focusing on quantum $2$-designs. These are arrangements of weighted  subspaces which are in a specific sense optimal for quantum state tomography.  After a brief introduction, we discuss the role of POVMs in quantum theory, developing  the importance of  quantum $2$-designs.  In the final section, we collect many if not most known examples of quantum-$2$ designs to date.\end{abstract}

\maketitle


\section{Introduction}

This survey is concerned with the role of frames and designs in quantum information theory.
Frames are spanning families in Hilbert spaces that permit stable expansions of vectors.
In contrast to orthonormal bases, frames can incorporate linear dependencies. The flexibility
in their design can be used to realize goals that would be unattainable with minimal systems,
Riesz bases. There is an intriguing connection between objects studied in frame theory and
in quantum information theory. A Parseval frame is a family of vectors $\{f_j\}_{j \in J}$
in a Hilbert space $\mathcal H$, indexed by an at most 
countable set $J$, such that the rank-one Hermitians obtained from the frame vectors
resolve the identity,
$$
   \sum_{j \in J} f_j \otimes f_j^* = I \, .
$$
Here, $f^*$ is the linear functional $x \mapsto \langle x, f_j \rangle, x \in \mathcal H$.
When we abbreviate the rank-one Hermitians appearing in this sum as  $A_j = f_j \otimes f_j^*$, then
the family $\{A_j\}_{j \in J}$  happens to be a special case of a positive operator-valued measure (POVM),
a family of positive (semidefinite) operators  that sum to the identity.
As explained in the next section,
POVMs describe the statistics for outcomes of quantum measurements, and the question
of selecting an optimal one to estimate states based on the observed relative frequencies for outcomes
is closely related to frame theory. For general quantum measurements, each $A_j$ does not have to
be a rank-one operator. Indeed, choosing each $A_j$ of a POVM to be a multiple of an orthogonal projection
 has also been the subject of frame theory under the name of fusion frames. In the remainder of this paper,
we study POVMs with more structure motivated by goals in quantum state tomography.

Following Zauner's seminal work \cite{Zauner2011}, we define a {\it weighted quantum design} to be a sequence
of orthogonal projections $\{P_j\}_{j=1}^n$ in a Hilbert space $\mathcal H$, accompanied by weights $\{w_j\}_{j=1}^n$ in $\mathbb R^+$.
The sequence is $t$-coherent for $t \in \mathbb N$ if for each unitary $U$ on $\mathcal H$,
$$
   \sum_{j=1}^n w_j P_j^{\otimes t} = \sum_{j=1}^n w_j ( U P_j U^*)^{\otimes t} \, .
$$
A weighted quantum design that is $s$-coherent for all integers $1 \le s \le t$ is called a quantum $t$-design.
Hence, a quantum $t$-design is also a $1$-design and satisfies, by averaging over $U$ with respect to the normalized Haar measure,
$\sum_{j=1}^n w_j P_j = \sum_{j=1}^n w_j \TR[P_j/d] I$, where $d$ is the dimension of $\mathcal H$. If 
the weights are scaled by an overall factor so that $\sum_{j=1}^n w_j \TR[P_j/d]  = 1$, then
such a quantum $t$-design is a POVM with a special structure.

\section{From quantum state tomography to quantum designs}

Quantum theory evolved in the early 20th century because results of physical measurements
could not be explained with models of classical mechanics \cite{Einstein1905a}.
Today, the predictive power of
quantum theory is unrivaled in its accuracy in physical experiments \cite{Levineetal1997}. 
In the last 30 years, quantum theory has
received much attention because of the computational advantages afforded by the intrinsic parallelism
in quantum time evolution \cite{Feynman1982}. 

The achievements of quantum theory come with conceptual and philosophical issues about
its interpretation. One of the main challenges is that the predictions are probabilistic.
In order to compare classical and quantum theories, a statistical framework is 
necessary. Moreover, because the outcomes of experiments are formulated in
language based on boolean logic, any viable quantum theory needs to include
the means to describe classical systems.

We briefly review the axioms of quantum theory, here with the simplifying assumption that the underlying Hilbert space 
$\mathcal H$ is finite dimensional: States are trace-normalized non-negative operators 
on this Hilbert space. Observables describe the information gained from observing the outcome of 
an experiment. The minimal view of quantum theory is that it predicts probabilities for observing outcomes.
A classical theory is embedded in the quantum theory by associating with
(mutually exclusive) outcomes of an experiment mutually orthogonal subspaces of
the Hilbert space. 

A measurement is performed by a {\em quantum instrument} \cite{MR3196987}. In order to keep track of the classical information 
gained from the measurement, we use the tensor product Hilbert space $\mathcal H \otimes \ell^2(\Omega)$, where $\mathcal H$
is used to define the quantum state and $\ell^2(\Omega)$ the outcomes of the experiment recorded by the measurement apparatus. 

Associating with each outcome $\omega \in \Omega$
the diagonal projection operator $D_\omega$ on $\ell^2(\Omega)$ whose range consists of the one-dimensional
subspace of functions supported on the singleton set $\{\omega\}$, a quantum instrument  is modeled by a trace-preserving map $\mathcal E$ of the form
$$
   \mathcal E( W  ) = \sum_{\omega \in \Omega} \sum_{j\in J_\omega} M_{j,\omega} W M_{j,\omega}^* \otimes D_{\omega} \, 
$$
where $W$ is a state operator on the Hilbert space $\mathcal H$, $D_\omega$ is a diagonal projection on $\ell^2(\Omega)$ while $\sum_{j \in J_\omega} M_{j,\omega} W M_{j,\omega}^*$
describes the non-normalized quantum state in $\mathcal H$ after the measurement with outcome $\omega$.

More informally, the outcome of a measurement is recorded, while the quantum state undergoing the measurement
is mapped by conjugating it with Kraus operators belonging to each outcome, thus creating a special form of a quantum channel.
A general quantum channel on a finite-dimensional Hilbert space has the form $\mathcal E: W \mapsto \sum_{i=1}^n {M_i} W M_i^*$
where the sequence of Kraus operators $\{M_i\}_{i=1}^n$ satisfies $\sum_{i=1}^n M_i^* M_i = I$. 

If one is only interested in the statistics of the outcomes, then these are given by the probabilities
$p_\omega = \sum_{j \in J_\omega} \TR[ M^*_{j,\omega} M_{j,\omega}  W ] $ obtained after performing a partial trace
over the Hilbert space describing the quantum state. If the probabilities for the outcomes determine each state uniquely, then
the measurement is called informationally complete. From now on, we assume that each outcome is associated with one Kraus operator
and abbreviate $A_\omega =M^*_\omega M_\omega$. The collection $\{A_{\omega}: \omega \in \Omega\}$ is then a 
family of positive
semi-definite operators summing to the identity, also known as a positive-operator valued measure.

\begin{defn}
A positive operator-valued measure $\{A_j \}_{j=1}^n$ is called informationally complete if
the map $\mathcal A: W \mapsto (\TR[A_j W])_{j=1}^n$ is injective on the set of operators $\{W: W\ge 0, \TR W =1\}$.
\end{defn}

Since $\mathcal A$ is a linear map, informational completeness is equivalent to the set $\{A_j \}_{j=1}^n$ having as its real span the space of all Hermitian matrices.

\begin{theorem}
A positive operator-valued measure $\{A_j\}_{j=1}^n$ is informationally complete
if and only if  the real span of $\{A_j\}_{j=1}^n$ consists of all Hermitian operators.
\end{theorem}
\begin{proof}
Assume the POVM is informationally complete and a Hermitian $H$ gives $\mathcal A(H)=0$. By summing entries of $\mathcal A(H)$, we get
$\sum_{j=1}^n \TR[A_j H] = \TR[H] = 0$, so $H$ has vanishing trace. 
We need to exclude $H \ne 0$.
Splitting $H=H_+ - H_-$, where
$H_+, H_- \ge 0$ yields $\TR[H_+] = \TR[H_-]$. If each of these traces is non-zero, then
we can assume $\TR[H_+]=\TR[H_-]=1$ without loss of generality, and thus
 $H_+$ and $H_-$ are  states that give $\mathcal A(H_+) = \mathcal A(H_-)$, so by 
 completeness of the POVM, $H_+=H_-$ and thus $H=0$. Since the kernel is trivial
 the range of the adjoint map $\mathcal A^*$ is the orthogonal complement, consisting of all Hermitians.
 However, the adjoint map is simply $\mathcal A^*(x) = \sum_{j=1}^n x_j A_j$ for $x \in \mathbb R^n$,
 so the POVM spans the space of Hermitians.
 
 Conversely, if the POVM is spanning, and $W_1, W_2$ are states with $W_1 \ne W_2$, then
 $H = W_1 - W_2 \ne 0$ and $\mathcal A(H) \ne 0$, so $\mathcal A(W_1) \ne \mathcal A(W_2)$.
 We conclude the POVM is informationally complete.
 \end{proof}
 
 The dimension of the real vector space of Hermitians on a $d$-dimensional complex Hilbert space
 is $d^2$. Because of the spanning property, this dimension is the minimum size required for informational completeness. 
 
 \begin{corollary}
 If a POVM $\{ A_j \}_{j=1}^n$ on a $d$-dimensional Hilbert space is informationally complete, then $n \ge d^2$.
 \end{corollary}

In finite dimensions, the spanning property is equivalent to $\{A_j \}_{j=1}^n$ being a frame for the inner product space of Hermitian matrices, equipped with the Hilbert-Schmidt inner product. The reconstruction of the state then proceeds by constructing the pseudo-inverse $H$
of the Gram matrix $G=(G_{j,l})_{j, l =1}^n$ with entries $G_{j, l } = \TR[A_j A_l]$. Let $\{B_j\}_{j=1}^n$ be the operator-valued canonical dual
of the POVM, so 
if $H=(H_{j,l})_{j, l =1}^n$ has entries $H_{j, l } = \TR[B_j B_l]$, then $GHG=G$ and $HGH=H$.
This implies that for each operator on $\mathcal H$, in particular each positive, trace-normalized $W$,
$$
   W = \sum_{j=1}^n B_j \TR[W A_j] \, .
$$
Thus, $W$ is expressed linearly in terms of the probabilities of observing outcomes.
Next, we consider the error induced by measuring relative frequencies $\hat p_j$ for outcomes instead of probabilities 
$p_j=\TR[WA_j]$. We quantify the error by the averaged squared Euclidean distance
between $W$ and the recovered state. In order to eliminate any bias, we randomize the 
measurement by conjugating $W$ with a unitary, selected uniformly at random.

\begin{defn}
Let $W$ be a state.
The mean-squared error associated with a POVM $\{A_j\}_{j=1}^n$ and the randomized input state $W^U\equiv U W U^*$ with the unitary $U$ chosen uniformly at random
is
$$
  \mathbb E[ \| W^U - \widehat W^U \|^2 ] = \mathbb E[ \sum_{j, l} (p^U_j - \hat p^U_j)(p^U_l - \hat p^U_l) \TR [B_j B_l]  ]\, .
$$
\end{defn}

The expression for the mean-squared error can be simplified substantially.
We follow Scott's analysis \cite{MR2269701} and assume, for simplicity, that only one outcome has been observed.
Consequently, $p_j^U = \TR[U W U^* A_j]$ and
$\{\hat p^U_j\}_{j=1}^n$ is a collection of $\{0,1\}$-valued random variables with 
$\mathbb E[ \hat p^U_j] = p^U_j$
and $\sum_{j=1}^n \hat p^U_j = 1$.

\begin{prop}[Scott \cite{MR2269701}]
With $W$, $\{A_j\}_{j=1}^n$ and $U$ as in the above definition, and a $\{0,1\}$-valued random probability vector $\hat p$
as described, 
$$ \mathbb E[ \| W^U - \widehat W^U \|^2 ]
=  \sum_j \TR[A_j/d] \TR[B_j^2] - \TR[W^2]  \, .
$$
\end{prop}
\begin{proof}
From the properties of $\hat p$,  $\mathbb E[ \hat p^U_j \hat p^U_l ] = p^U_j \delta_{j,l} $.
Consequently,
$$
 \mathbb E[ \| W - \widehat W \|^2 ] = \mathbb E[\sum_{j, l} (p^U_j\delta_{j,l} -p^U_j p^U_l) \TR[B_j B_l]
 = \sum_j \mathbb E[p^U _j ]\TR[B_j^2] - \TR[W^2] \, . 
$$
Finally, averaging with respect to $U$, we get
$\mathbb E[p_j^U] = \mathbb E[ \TR[U W U^* A_j]] = \TR[A_j]/d$.
\end{proof}

Hence, the optimal measurement for the mean-squared error minimizes the first term,
while the second term remains state dependent. Scott formulated a lower bound for the first term \cite{MR2269701}
and characterizes cases of equality as $\{A_j\}$ being related to weighted complex projective 2-designs.

\begin{defn}
A sequence of rank-one projection operators $\{\pi_j\}_{j=1}^n$ with weights $\{w_j\}_{j=1}^n$
is a weighted projective 2-design if 
$$
   \sum_{j=1}^n w_j \pi_j \otimes \pi_j = \frac{2}{d(d+1)} \Pi_{\mathrm{sym}}
$$
where $ \Pi_{\mathrm{sym}}$ is the projection onto the symmetric subspace of $\mathcal H \otimes \mathcal H$,
meaning
$$
  \Pi_{\mathrm{sym}} = \frac 1 2 \sum_{j,k=1}^d (E_{j,j} \otimes E_{k,k} + E_{j,k}\otimes E_{k,j} ) \, 
$$
in terms of matrix units $E_{j,k} = e_j \otimes e_k^*$.
\end{defn}

\begin{thm}[Scott \cite{MR2269701}]
Given a POVM $\{A_j\}_{j=1}^n$ on a $d$-dimensional Hilbert space $\mathcal H$ 
and the operator-valued canonical dual $\{B_j\}_{j=1}^n$, then
$$
   \sum_j \TR[A_j/d] \TR[B_j^2] \ge \frac 1 d + (d^2-1)(d+1)/d \, .
$$
Moreover, equality holds in this inequality if and only if $\{A_j\}_{j=1}^n$
is a rank-one POVM and $w_j= \TR[A_j]/d$, $\pi_j = A_j / \TR[A_j]$ forms a weighted projective 2-design. 
\end{thm}

\begin{proof}
In order to find a lower bound for the first term, Scott notices that it is the trace 
of an operator-valued frame operator belonging to the frame $\mathcal Q$ with vectors $Q_j =\sqrt{\TR[A_j/d]} B_j$.
The eigenvalues of the frame operator of $\mathcal Q$ are controlled by its operator-valued canonical dual
$\mathcal P$ with vectors $P_j = \sqrt{d/\TR[A_j]} A_j$.
To begin with, we note
$$
   \sum_{j=1}^n P_j \TR[I P_j] = d \sum_{j=1}^n A_j = d I 
$$
so one of the eigenvalues of $\mathcal P$  equals $d$, hence $\mathcal Q$ has an eigenvalue $1/d$.

We also observe 
$\TR[P_j^2] = d \TR[ A_j^2]/\TR[A_j] \le d \TR[ A_j]$ because $0 \le A_j \le I$, so $\TR[A_j^2] \le (\TR[A_j])^2$.
Hence, $\sum_{j=1}^n \TR[ P_j^2] \le d \sum_{j=1}^n \TR[ A_j ] = d^2$.

Since the eigenvalues of $\mathcal Q$ are inverses of those of $\mathcal P$, we see that
we want to optimize $\sigma(\lambda) \equiv \sum_{j=1}^{d^2} \lambda_j^{-1}$ subject to $\sum_{j=1}^{d^2} \lambda_j \le d^2$
and $\lambda_1 = d $. By the convexity of  the function $x \mapsto 1/x$, the minimum of $\sigma$
cannot be smaller than for the choice $\lambda_1=d$ and $\lambda_j = (d^2-d)/(d^2-1)=d/(d+1)$,
so
$$
   \sigma(\lambda) = \frac 1 d + (d^2-1)(d+1)/d
$$
Equality is achieved if and only if  each $A_j$ is of rank one and 
$\mathcal P$ has two eigenvalues (ignoring multiplicity), $\lambda_1=d$ and
$\lambda_j = d/(d+1)$. Consequently, each $\pi_j = A_j /\TR[A_j]$ is a rank-one projection
and for each Hermitian $H$ with vanishing trace,
$$
     \sum_{j=1}^n A_j \TR[ H A_j] / \TR[A_j] = \frac{1}{d+1} H \, .
$$ 

Finally, we use the characterization of weighted projective 2-designs: Assuming each $A_j$ is rank one, let $w_j = \TR[A_j]/d$ and $\pi_j = A_j / \TR[A_j]$, then
$$
  \sum_{j,l=1}^n w_j w_l |\TR[ \pi_j \pi_l]|^2 = \frac{1}{d^4} \sum_{j,l=1}^n |\TR[ P_j P_l ] |^2
$$
The sum is identified as the squared Hilbert-Schmidt norm of the frame operator of $\mathcal P$,
so
$$
  \sum_{j,l=1}^n w_j w_l |\TR[ \pi_j \pi_l]|^2  = \frac{1}{d^4} ( d^2 + (d^2-1)d^2/(d+1)^2) = \frac{2}{d(d+1)} \, .
$$
This last identity is equivalent to $\sum_{j=1}^n w_j \pi_j \otimes \pi_j = \frac{2}{d(d+1)} \Pi_{\mathrm{sym}}$
\cite{MR2269701}[Theorem 6], hence equality holds in the bound if and only if the weights and rank-one orthogonal projections form a weighted complex projective 2-design.
\end{proof}
%

We conclude that $\{A_j\}_{j=1}^n$ meets the bound for the first term in the mean-squared error
if and only if
the weighted projections associated with the POVM
form a weighted complex projective 2-design. 

\begin{corollary} Let $\{A_j\}_{j=1}^n$ be a POVM, then
the mean-squared error for linear reconstruction from observing one outcome is bounded below by
$$
 \mathbb E[ \| W^U - \widehat W^U \|^2 ] \ge \frac 1 d + (d^2-1)(d+1)/d - \TR[W^2]
$$
and equality holds if and only if each $A_j$ in the POVM is rank one and $w_j=\TR[A_j/d]$, $\pi_j = A_j / |TR[A_j]$
form a weighted complex projective 2-design.
\end{corollary}

To summarize, if weighted complex projective 2-designs exist, then they are optimal for quantum state
tomography with linear reconstruction using the operator-valued canonical dual of the associated POVM.

\input{knowndesigns.tex}

\bibliography{universal}
\bibliographystyle{plain}

\end{document}

%% file: knowndesigns.tex
\section{Known quantum $2$-designs}
We conclude by collecting known examples of quantum designs, highlighting connections between the design properties and geometric properties.  We briefly sketch history and construction principles, occasionally demonstrating concrete examples and highlighting open problems along the way.

The presentation  is organized as follows.  In Sections~\ref{levdesigns} and~\ref{weighted}, we examine quantum $2$-designs comprised of rank one projectors,  encoding the elements of each design as the columns of unit norm tight frames,
$$
F=[f_1 \, f_2 \, \dots \, f_n],
$$
from which the POVM's projection arises via
$$
P_j := f_j \otimes f_j^*, j \in \{1,.,2,\dots, n\}.
$$  

In the first case, Section~\ref{levdesigns},
we examine equally weighted complex projective $2$-designs, a special class of POVMs whose quantum $2$-design property is characterized by minimal {\it coherence},
$$ \mu(F)= \max_{j\neq l} |\langle f_j, f_l\rangle|,$$
in particular achieving the lower bound given by Levenstein~\cite{Levenshtein1992, lev:2017, levinprep} (see Theorem~\ref{levbd} below).  In Section~\ref{weighted}, we examine quantum designs manifesting as weighted complex $2$-designs, for which minimal coherence - as verified in some cases~\cite{MR3557826} - is a byproduct, but not a necessary condition for the the quantum $2$-design property.  In these cases, the $2$-design property is characterized according the weighted framework of Roy and Scott~\cite{RoyScott2007} (see Theorem~\ref{th_wtd}).
Finally, in Section~\ref{higherrank}, we point out families of quantum $2$-designs manifesting whose values are projections of rank greater than one.

\subsection{Rank one complex projective $2$-designs and Levenstein's bound}\label{levdesigns}
The  quantum $2$-design property of the three families in this section are characterized by Levenstein's 2nd degree coherence bound~\cite{Levenshtein1992, lev:2017, levinprep}. 

\begin{theorem}\label{levbd}[Levenstein; \cite{Levenshtein1992}, see also~\cite{lev:2017, levinprep}]
	Given a unit norm frame $\Phi=\{\phi_j\}_{j=1}^n \subset \mathbb C^d$, then the corresponding rank one operators $\Phi$ form an equally weighted complex projective $2$-design (which is at most a $3$-design) if and only if the following lower bound on frame coherence is met with equality:
$$
\mu(\Phi) \geq \sqrt{\frac {2n - d(d+1)}{(n-d)(d+1)}}.
$$
\end{theorem}

\subsubsection{SIC-POVMs and Zauner's conjecture}

Among all quantum designs, {\it Symmetric Informational Complex Positive Operator-valued Measures (SIC-POVMs)} seem to be the ones of prime interest to academians for several reasons. Foremost, a given SIC-POVM, $\{P_j\}^{d^2}$ for $\mathbb C^d$, is minimal in the sense that $d^2$ is just enough projections to span the space of matrices, which corresponds to the cheapest possible model for optimal measurement apparatuses.  Moreover, a SIC-POVM $\Phi$ admits an appealing geometry; in particular, they correspond to maximal equiangular tight frames -  meaning they are the frames of largest admissible cardinality, ie $n=d^2$,  such that the absolute inner products between distinct elements is constantly equal to Welch's lower coherence bound~\cite{Welch1974}:
        $$
             \mu(\Phi)\geq \sqrt{\frac{n-d}{d(n-1)}}.
        $$

  In fact, maximal ETFs are precisely where Levenstein's bound and Welch's bound agree\cite{lev:2017}.  Setting $n=d^2$, both bounds evaluate as follows:
 $$\sqrt{\frac{2d^2 -d(d+1)}{(d^2-d)(d+1)}}=\sqrt{\frac{d^2-d}{d(d^2-1)}} = \frac{1}{d+1}
 $$

SIC-POVMs have enjoyed a rich history dating back to the dissertation of Zauner~\cite{Zauner1999} (see also~\cite{Zauner2011} for an English translation), wherein he conjectured the existence of a SIC-POVM in $\mathbb C^d$ for every $d \in \mathbb N$. In fact, his conjecture is more constructive than existential than one might infer from the preceding statement.  In particular, for each $d$, Zauner anticipates that a SIC-POVM may be generated under the Weyl-Heisenberg group, $\mathbb H_d$, by some rank one operator, X -- called a {\it fiducial vector} -- where, in addition, $X$ commutes with some element, $T$, of the underlying Jacobi group $\mathbb H_d \rtimes \mathbb{SL} (\mathbb Z_2,d)$ and where the orbit under $T$ on $\mathbb H_d \backslash Z(\mathbb H_d)$ consists of three elements.

To date, analytic/symbolic verifications of the predicted fiducial vectors - each satisfying the structure predicted by Zauner - have been determined for $d=1,\dots,21,24,28,30,31,35,$
\\
$37,39,43$ and $48$~\cite{Zauner1999, MR2059685, MR1863709, ScottGrassl2010, chienthesis:2015, flammiawebsite, 2017arXiv170305981A}; moreover, due to recent advances, numerical approximations of fiducial vectors seemingly in agreement with Zauner's conjecture have been verified for $d \leq 150$~\cite{MR2059685, 2016arXiv161207308F, 2017arXiv170303993S, 2017arXiv170307901F}.
    An existential or constructive proof that Zauner's conjecture holds in infinitely many dimensions remains an outstanding open problem.
    
       \begin{ex}[A SIC-POVM for $C^2$]
    	With $\omega$ a primitive third root of unity, the columns of the $2 \times 4$  unit norm frame,
    	$$F= \left[
    	\begin{array}{cccc}
    	1 & \frac{1}{3} & \frac{1}{3} & \frac{1}{3}\\
    	0 & \frac{4}{9} & \frac{4}{9} \omega & \frac{4}{9} \omega^2\\
    	\end{array}
    	\right],
    	$$
    	 correspond to a maximal equiangular tight frame, whence generating a SIC-POVM for $\mathbb C^2$,
    	$$\mathcal P = \left\{ f_j \otimes f_j^* \right\}_{j=1}^4.$$  This SIC-POVM manifests as a tetrahedron on the Bloch sphere.    
    \end{ex}

\subsubsection{Maximal sets of mutually unbiased bases}
A pair  of {\it mutually unbiased bases (MUBs)}, $B$ and $B'$, is a pair of orthonormal bases for $\mathcal H$ 
with the property that $|\langle x, y \rangle |^2 = \frac 1 d $ for every $x \in B$ and $y \in B'$.  
It is known~\cite{DelsarteGoethalsSeidel1975} that if $\mathcal B$ is a family of pairwise mutually unbiased bases for $\mathbb C^d$, then $\left| \mathcal B \right| \leq d+1$.  
In the case that a maximal family of mutually unbiased bases $B_1, B_2, \dots, B_{d+1}$ exists, then their union, written as a block matrix, $$F=[B_1 \, B_2 \, \dots \, B_{d+1}],$$
 consists of $n=d(d+1)$ unit vectors, which together achieve  the orthoplex bound~\cite{Rankin1955, ConwayHardinSloane1996, MR3557826}, 
     $$
         \mu(\Phi) \geq \frac{1}{\sqrt{d}},
     $$
 a lower bound on frame coherence which holds whenever $n>d^2$.  Indeed, it is known~\cite{lev:2017} if a unit norm frame, $\Phi$, achieves the orthoplex bound, then $\left| \Phi \right| \leq d(d+1)$. Consequently, maximal families of MUBs are also maximal unit norm frames that achieve the orthoplex bound, and the converse, that maximal unit norm frames that achieves the orthoplex bound necessarily form families of MUBs, is also true~\cite{1523643}. Interestingly, as with the Welch bound, it is only when a maximal family of MUBs exist that Levenstein's bound  agrees with  the orthoplex bound.  Setting $n=d(d+1)$, we see
$$
      \sqrt{
      	\frac{2d(d+1) - d(d+1)}
             {(d(d+1)-d)(d+1)}}
  =
      \frac{1}{\sqrt d}.
$$

Unlike SIC-POVMs, maximal sets of MUBs are a class $2$-designs currently enjoying the certainty of existence in an infinitude of dimensions. Thanks to the work of~cite{CameronSeidel1973}, we know that than maximal sets of MUBs exist whenever $d$ is a prime number, and the authors of~\cite{WoottersFields1989} later extended the result to all prime powers.  Maximal MUBs of prime power order enjoy several straightforward combinatorial constructions, for example, using mutually orthogonal squares\cite{Rao2010} or relative differences~\cite{CASAZZA2017}.  Nevertheless, upon consideration of the first composite dimension, $d=6$, one encounters a well-known open problem:are there any composite dimensions, $d$, for which $d+1$ MUBs exist?  Numerical evidence suggests that no more than three pairwise MUBs coexist in $\mathbb C^6$~\cite{Jaming2010}, but the answer remains unknown.

\begin{ex}[Three MUBs in $\mathbb C^2$]
	The columns of matrices, 
	$$B_1 = \left[\begin{array}{cc}
	1 & 0 \\
	0 & 1 \\
	\end{array} \right], \, \, \, \, 
	B_2 = \frac{1}{\sqrt 2}\left[\begin{array}{cc} 1 & 1 \\ 1 & -1 \end{array}\right]
	\text{ and }
	B_3 = \frac{1}{\sqrt 2}\left[\begin{array}{cc} i & i \\ i & -i\end{array}\right],
	$$
	correspond to a set three MUBs in $\mathbb C^2$, whence a quantum $2$-design of six rank one projections.  The six points of this example manifests as the vertices of an octahedron on the Bloch sphere.
\end{ex}

\subsubsection{Sporadic Levenstein $2$-designs}
Besides SIC-POVMs and maximal sets of MUBs, the authors of~\cite{lev:2017} distilled, via integrality analysis, that there are sparsely many feasible pairs, $(d,n)$, where non-equiangular, non-orthoplectic Levenstein-equality frames exist.  Indeed, to date, there are exactly four known nonequiangular, nonorthoplec unit norm frames which achieve Levenstein's 2nd degree bound, thereby forming quantum $2$-designs in locations where the optimal coherence is characterized neither by the Welch bound nor orthoplex bounds.  We list these four examples in Table~\ref{tbl_four}, along with several other feasible pairs -- up to $d=28$ -- where Levenstein-equality frames might exist based on sufficiency of certain integrality relations.  In the table, $\alpha$ denotes the squared Levenstein bound for the given pair, $(d,n)$, and question marks indicate that no construction of feasible the $2$-design corresponding to the given row is currently known.   

\begin{table}
	\centering
	\begin{tabular}{|c|c|c|c|c|c|}
		\hline
		$d$ & $n$ &  $1/\alpha$ & notes & location \\  \hline
		4 & 40 & 3 &  Witting polytope~\cite{MR1119304} & Hoggar~\cite{Hoggar1982}\\
		5 & 45 & 4 &  Shephard--Todd no.\ 33 & Hoggar~\cite{Hoggar1982}\\
		6 & 126 &  4 & Shephard--Todd no.\ 34 & Hoggar~\cite{Hoggar1982}\\
		7 & 112 & 5 & ? & ?\\
		8 & 120 & 6 & ? & ? \\
		8 & 288 &  5 & ? & ? \\
		9 & 225 & 6 & ? & ? \\
		10 & 220 & 7 & ? & ?\\
		10 & 550 & 6 & ? & ?\\
		11 & 176 & 9  & ? & ? \\
		11 & 231 & 8  & ? & ?\\
		\vdots & \vdots & \vdots & \vdots & \vdots\\
		28 & 1624 & 19 & ? & ?\\
		28 & 4060 & 1755 & Rudvalis group~\cite{MR0335620} & Hoggar~\cite{Hoggar1982}\\
		\hline  
	\end{tabular}
\caption{Admissible Levenstein equality pairs in low dimensions}
\label{tbl_four}
\end{table}

\subsubsection{Rank one Levenstein-type quantum designs with higher design parameters}
Before moving to the next section, where we consider weighted rank one POVMs, we note that Levenstein provided an infinite family of lower coherence bounds by applying linear programming techniques to so called special orthogonal polynomials of increasing degrees.  Up until now, we have only considered his second degree bound (see Theorem~\ref{levbd}), as, unfortunately, it seems that examples that achieve these higher degree bounds are extremely rare, see for example~\cite{Hoggar1989}. Nevertheless, any example that achieves one of these higher degree bounds will form a complex projective $t$-design with a correspondingly higher parameter, $t$, thereby forming a quantum $2$-design.  Presently, we are aware of only one example, a family of six orthonormal bases for $\mathbb C^2$, whose union forms a projective $5$-design, thus yielding a quantum $2$-design of $12$ projectorss~\cite{2016arXiv160502012C}.  This example manifests as the vertices of a regular icosahedron when embedded into the Bloch sphere.

\subsection{Weighted rank one complex projective $2$-designs.}\label{weighted} 
Next, we consider those quantum $2$-designs which manifest, more generally, as 
weighted complex projective $2$-designs.
SIC-POVMs and maximal sets of MUBs, for example, are equally weighted complex projective $2$-designs,  with each $w_j = \frac 1 n$.
The following useful theorem characterizes weighted designs  in terms of their underlying  geometry.

\begin{theorem}\label{th_wtd}[Roy/Scott; \cite{RoyScott2007}]
	A $d\times n$ unit norm tight frame $F=[f_1 \, f_2 \, \dots \, f_n]$, along with normalized weights, $\{w_j\}_{j=1}^n$, $\sum_{j=1}^n w_j=1$,
	forms a complex weighted $2$-design with rank-one projections $\pi_j = f_j \otimes f_j^*$, if and only if
	$$
	\sum\limits_{j=1}^{n}    \sum\limits_{k=1}^{n} w_j w_{k}  \left|\langle f_j, f_{k}  \rangle   \right|^{2t}
	=
	\left(
	\begin{array}{cc}
     d + t -1 \\ t \end{array}
	\right)^{-1}.
	$$ 
\end{theorem}

\subsubsection{Weighted orthonormal bases}
In demonstration of Theorem~\ref{th_wtd}, given that $d=p^s$ or $d+1=p^s$ for some prime $p$ and $s\in \mathbb N$ the authors ~\cite{RoyScott2007} use nonlinear functions on nonabelian groups to construct $m$ families of weighted orthonormal bases, where $m>d+p+1$. Moreover, they show that whenever $d+1=p$ is a prime power, then a weighted complex projective $2$-design comprised of $d+2$ MUBs exist~\cite{RoyScott2007}.  For example, while there is no known construction of seven equally weighted MUBs in $\mathbb C^6$, their weighted construction~\cite{RoyScott2007} forms quantum $2$-design of $8$ weighted MUBS -- or $48$ rank one orthogonal projections -- since $d=6$ is indeed a prime power plus one.

\subsubsection{Weighted unit norm tight frames}
Using the idea of weighting from the previous approach, the authors of~\cite{MR3557826} relax the requirement that a weighted design manifest as a set of bases, instead considering sets of unit norm tight frames, noting that such a union corresponds to automatically forms a projective $1$-design.  By exploiting DFT matrices and combinatorial objects called difference sets, the authors produce families of equiangular tight frames and biangular tight frames, which when appropriately weighted and adjoined with canonical orthonormal bases, form quantum $2$-designs.  More precisely, they produced two infinite families.  The first family consists of $n=d^2+1$ rank one projections for $\mathbb C^d$ whenever $d$ is a prime power plus one, and the second yields  quantum $2$-designs of $n=d^2+d-1$ rank one projectors whenever $d$ is a prime power.

\begin{ex}
	We may think of the standard orthonormal basis and any properly scaled selection of two rows from a $3 \times 3$ DFT matrix as unit norm frames for $\mathbb C^2$.   The union of such frames,
	  $$F_1 = 
	  \left[
	  \begin{array}{cc} 
	  1 & 0 \\ 0 & 1 \end{array} 
	  \right] 
	  \text{ and }
	 F_2 = 
	 \left[
	 \begin{array}{ccc} 
	 1 & 1 & 1 \\ 1 & \omega & \omega^2 
	 \end{array} 
	 \right],$$
	 along with weights corresponding to the matrix columns, $w_1=w_2=1/6$ and $w_3=w_4=w_5=1/12$, form a weighted complex projective $2$-design.  
\end{ex}

A striking by-product of these examples is that they happen to form minimally coherent frames, in particular achieving the orthoplex bound.  In general, there is no reason to expect a weighted $2$-design to have this property.    

\subsection{Quantum $2$-designs with high rank projections}\label{higherrank}
\subsubsection{Zauner's conjecture revisited}
In \cite{Appleby2007}, Appleby shows the existence of quantum $2$-designs of arbitrary but equally ranked values, which behave like SIC-POVMs.  By taking a random simplex and shrinking it appropriately to ``fit'' within the so-called Bloch space~\cite{Appleby2007}, he shows that for every $d$, there is some sub-dimension $l$ with $1\leq l \leq d-1$, and a  quantum $2$-design consisting of $n=d^2$ rank $l$ orthogonal projections.  By ``behave like SIC-POVMs'', we mean these objects provide the desired quantum $2$-design properties and they are symmetric in the sense that their set of pairwise Hilbert Schmidt inner products is a singleton, ie they are equiangular.

\subsubsection{Maximal sets of mutually unbiased subspaces}
In his thesis, Zauner mentions this example of quantum $2$-designs~\cite{Zauner1999}.  Given an ambient Hilbert space with dimension of even prime power, $d=2^t$, $t>1$, one may construct a family of $n=2^{t+1}-2$ orthogonal projections with rank $l=2^{t-1}$ (ie, `half-dimensional' subspaces) which constitute quantum $2$-designs.  In~\cite{2016arXiv160704546B}, the authors provide a concrete construction of this example by exploiting maximal families of MUBs and maximal sets of Johnson codes with optimal distance relations.